\documentclass[smallcondensed]{svjour3}
\usepackage{latexsym}
\usepackage{graphics}
\usepackage{epsfig}
\usepackage{amsmath}
\usepackage{amssymb}
\usepackage{fancybox}
\usepackage{bm}
\usepackage{float}
\usepackage{hyperref}

\usepackage{mathrsfs}
\usepackage{fancyhdr}
\usepackage{ascmac}

\def\cdott{{\cdot\cdot}}

\def\bullett{{\bullet\bullet}}

\def\astt{{\ast\ast}}

\def\ooo{\mathfrak{o}}

\def\Z{\mathbb{Z}}

\def\TTT{{\mathfrak T}}

\def\td{d_\www}

\def\vvv{\mathfrak{v}}

\def\FFF{\mathfrak{F}}

\def\sss{\mathfrak{s}}

\def\SSS{{\mathfrak S}}

\def\LLL{{\mathfrak L}}

\def\aaa{{\mathfrak a}}
\def\AAA{{\mathfrak A}}

\def\bbb{{\mathfrak b}}

\def\CCC{{\mathfrak C}}

\def\R{\mathbb{R}}

\def\Varepsilon{{\mathcal E}}

\def\eee{{\mathfrak e}}

\def\iii{{\mathfrak i}}
\def\I{\mathscr{I}}

\def\M{{\cal M}}
\def\MM{\mathscr{M}}

\def\R{\mathbb{R}}

\def\RRR{\mathfrak{R}}

\def\vomega{{\mathfrak w}}
\def\www{{\mathfrak w}}
\newcommand{\QED}{\nobreak \ifvmode \relax \else
      \ifdim\lastskip<1.5em \hskip-\lastskip
      \hskip1.5em plus0em minus0.5em \fi \nobreak
      \vrule height0.75em width0.5em depth0.25em\fi}
\begin{document}
\title{
Characteristic Classes in General Relativity
}
\subtitle{
on a Modified Poincar\'{e} Curvature bundle
}
\author{Yoshimasa Kurihara}
\institute{Y. Kurihara \at The High Energy Accelerator Organization (KEK), Tsukuba, Ibaraki 305-0801, Japan\\yoshimasa.kurihara@kek.jp}
\date{}
\maketitle
\begin{abstract}
Characteristic classes in space-time manifolds are discussed for both even- and odd-dimensional spacetimes. 
In particular, it is shown that the Einstein--Hilbert action is equivalent to a second Chern-class on a modified Poincar\'{e} bundle in four dimensions.
%The Einstein equation itself can be understood as a topological invariant.
Consequently, the cosmological constant and the trace of an energy-momentum tensor become divisible modulo $\R/\Z$.
\end{abstract}

%%%%%%%%%%%%%%%%%%%%%%%%%%%%%%%%
% Section 1: Introduction
%%%%%%%%%%%%%%%%%%%%%%%%%%%%%%%%
\section{introduction}\label{intro}
General relativity can be formulated as a Chern-Simons topological theory in a $(1+2)$-dimensional spacetime\cite{Witten198846}.
Gravitational theories based on Chern--Simons forms can be extended to any odd-dimensional spacetime\cite{0264-9381-29-13-133001}, allowing the gravitational constant to be discretized as modulo $\R/\Z$\cite{PhysRevD.51.490}.
In contrast, the Chern-Simons theory cannot be used for even-dimensional manifolds because it is identically equal to  zero in even dimensions.
However, even if the Chern--Simons form cannot be used to construct a four-dimensional gravitational theory, it may be possible to define general relativity using topological invariants, such as Chern classes, which can be defined in any dimension.
In reality, there is no known action defined by Chern classes for general relativity in even-dimensional spacetime, despite the efforts made to find such an action\cite{Kibble:1961ba,PhysRevLett.33.445,PhysRevD.13.3192,PhysRevLett.38.739,Izaurieta2009213}.  
These systematic failures are mainly due to the lack of translational symmetry of the gravitational action.

While the Einstein--Hilbert action of gravity is invariant under a global general linear group, $GL(1,3)$, and a local Lorentz group, $SO(1,3)$, it is not invariant under the local Poincar\'{e} group $ISO(1, 3)$.
This study introduced a modified translational operator and shows the four-dimensional Einstein--Hilbert action of pure gravity is invariant under this new translation\cite{2017arXiv170305574K}.
In this report, topological classes on principal bundles are discussed for a large variety of gravitational theories in even-dimensional manifolds.
In particular, a Lagrangian constructed using a second Chern-class in a four-dimensional manifold is discussed.
%%%%%%%%%%%%%%%%%%%%%%%%%%%%%%%%
% Section 2: Lovelock theory
%%%%%%%%%%%%%%%%%%%%%%%%%%%%%%%%
\section{Lovelock theory}\label{Lovelock}
Let $\MM$ be a $d$-dimensional global (pseudo-Reimannian) manifold with a metric tensor $g_{\mu\nu}$. 
It is assumed that a local Poincar\'{e} manifold, $\M$, can be assigned at any point on $\MM$.
Here, the local Poincar\'{e} manifold refers to a smooth manifold with symmetry under the Poincar\'{e} group $ISO(1,d-1)=SO(1,d-1)\ltimes T^d$, where $T^d$ is a $d$-dimensional translation group.
Herein, a manifold with the global $GL(1,d-1)$ and local $ISO(1,d-1)$ groups is called a ``space-time manifold''.
We define the metric tensor on $\M$ as $\eta_{ab}=diag(1,-1,\cdots,-1)$.
A map of  metric tensors  from the global space-time manifold to the local Poincar\'{e} manifold is obtained using a vielbein\footnote{A term ``vielbein'' is used in any dimensions, including a four-dimensional spacetime.} $\Varepsilon^\mu_a(x)$, such as
$
g_{\mu_1\mu_2}~\Varepsilon^{\mu_1}_{a}~\Varepsilon^{\mu_2}_{b}
=\eta_{ab}
$.
Here, we employ the Einstein convention for repeated indexes.
Moreover, we introduce a convention that the Greek- and  Roman-letter indexes represent the global and  local coordinates, respectively.
Both indexes start from zero to $d-1$.
The existence of the inverse function $(\Varepsilon^{\mu}_{a})^{-1}=\Varepsilon_{\mu}^{a}$, is assumed.
The one-form, referred to as ``vielbein form'' is introduced as $\eee^a=\Varepsilon^a_\mu dx^\mu$ on $T^*\M$.
Fraktur letters indicate differential forms throughout this report.
A ``spin form'' is defined using the spin connection $\omega_{\mu~b}^{~a}$ as
$
\vomega^{ab}=dx^\mu\omega^{~a}_{\mu~c}~\eta^{cb}
=\vomega^{a}_{~c}~\eta^{cb},
$ 
which satisfies $\vomega^{ab}=-\vomega^{ba}$.
For the above local representations, we consider orthonormal bases as $\partial_\mu=\partial/\partial x^\mu$ and $dx^\mu$ for the global frame bundle of $T_x\M$ and $T_x^*\M$,  respectively, at  $x\in\MM$.
Fundamental forms, the vilbein and spin forms defined on the local space-time manifold, are independent of the choice of a global frame vector.

The generators $\iii\sss\ooo(1,d-1)$ are denoted as $(J_{ab},P_a)$, where $J_{ab}$ and $P_a$ corresponds to $SO(1,d-1)$ and $T^d$, respectively.
The Lie algebra expressed as follows:
\begin{eqnarray}
\left[P_a,P_b\right]&=&0,\label{PJcr1}\\
\left[J_{ab},P_c\right]&=&-\eta_{ac}P_b+\eta_{bc}P_a,\label{PJcr2}\\
\left[J_{ab},J_{cd}\right]&=&
-\eta_{ac}J_{bd}+\eta_{bc}J_{ad}
-\eta_{bd}J_{ac}+\eta_{ad}J_{bc}.\label{PJcr3}
\end{eqnarray}
A Poincar\'{e} connection bundle (principal bundle) over $\MM$ can be expressed as follows:
\begin{eqnarray}
\AAA_0&=&\{J_{ab}, P_c\}\times\{\vomega^{ab},\eee^c\}
=J_\cdott\vomega^\cdott+ P_\cdot\eee^\cdot,\label{con}
\end{eqnarray}
where $\AAA_0$ is a connection-valued one-form.
Note that the vielbein form $\eee^\bullet$ functions as a connection in the Poincar\'{e} bundle. 
Hereafter, dummy {\bf Roman} indexes are omitted in the expressions when the index pairing is obvious, and a dot (or an asterisk) is used, as used in (\ref{con}).
Poincar\'{e} curvature forms can be defined from the connection forms as follows:
\begin{eqnarray}
\FFF_0&=&d\AAA_0+\AAA_0\wedge\AAA_0=J_\cdott\RRR^\cdott +P_\cdot\TTT^\cdot,\label{curvf}
\end{eqnarray}
where 
$\TTT^{a}=\td\eee^a=d\eee^{a}+\www^a_{~\cdot}\wedge\eee^{\cdot}$,
$\RRR^{ab}=d\www^{ab}+\www^a_{~\cdot}\wedge\www^{\cdot b}$,
and $\td$ is a covariant  derivative with respect to a local $SO(1,d-1)$.
%The symmetry $J_{ab}=J_{ba}$ is used above.
Here, the two-forms $\TTT^{a}$ and $\RRR^{ab}$ are called the torsion and curvature forms, respectively.

On the space of the Poincar\'{e} curvature bundle, one can introduce topological classes as series expansions of the curvature form.
\begin{definition}{\bf (Topological classes)}:
Let $\AAA$ and $\FFF=d\AAA+\AAA\wedge\AAA$ be the connection and curvature forms, respectively, on the Poincar\'{e} curvature bundle in a local $d$-dimensional space-time manifold $\M$.
\begin{enumerate}
\item {\bf Chern class}:
\begin{eqnarray}
\sum_{j=0}c_j(\FFF)t^j
&=&det\left(
I+i\frac{\FFF}{2\pi}t
\right)
\end{eqnarray}
can be defined globally, where $t$ is a real parameter, $I$ is a unit matrix, and $c_j(\FFF)$ is a $2j$-form, which is referred as the $j$'th Chern class.
\item {\bf Pontryagin class}:
\begin{eqnarray}
\sum_{j=0}p_i(\FFF)t^j
&=&det\left(
I+\frac{\FFF\wedge\FFF}{4\pi^2}t
\right).
\end{eqnarray}
$p_j(\FFF)$ is a $4j$-form, which is referred to the $j$'th Pontryagin class.
It obeys the relation $p_j(\FFF)=(-1)^jc_{2j}(\FFF)$.
\item {\bf Euler class}: The Euler class $\Varepsilon_n(\FFF)$ can be defined using a Pontryagin class as follows:
\begin{eqnarray}
\Varepsilon_n(\FFF)\wedge\Varepsilon_n(\FFF)&=&p_{n}(\FFF),
\end{eqnarray}
where the dimensionality is $d=2n$.
Thus, a Euler class can be defined only on an even-dimensional space-time manifold.
\end{enumerate}
\end{definition}
It is known that a cohomological group is associated with each characteristic class.
In particular, for the Chern classes, the following theorem has been proposed\cite{zbMATH03077491}:
\begin{theorem}{\bf{(Chern--Wiel theory)}}\label{CW}
\begin{enumerate}
\item Each Chern class, $c_n(\FFF)$, has a closed form. 
Thus these classes define the de Rham classes.
\item The difference between two Chern classes obtained from different connection bundles, $c_n(\FFF)-c_n(\FFF')$, gives an exact form.
Thus  the difference defines the same de Rham cohomology class (a de Rham class of a Chern form is independent of connections).
\end{enumerate}
\end{theorem}
This theorem is proven in previous studies\cite{roe1999elliptic,frankel_2011}.

A gravitational theory can be constructed for the global manifold $\MM$ by introducing an action integral that is invariant under both the global $GL(1,d-1)$ and local $SO(1,d-1)$ groups.
The presence of invariance under $T^d$ would prevent the general relativity from being constructed.
Hence, invariance is not required in this case.
For a possible form of the gravitational actions, Lovelock\cite{Lovelock:1971yv} and Zumino\cite{ZUMINO1986109} have proved the following remark:
\begin{remark}{\bf (Lovelock--Zumino action)}:
In a torsionless manifold, the most general gravitational action has the following form:
\begin{eqnarray}
\I(\eee,\RRR)&=&\kappa^{-1}\int_\M \sum_{p=0}^{[d/2]} a_p \LLL_p,\label{LL1}
\end{eqnarray}
where $\kappa$ is a gravitational constant and $a_p$ are arbitrary constant, using a (so-called) Lovelock Lagrangian:
\begin{eqnarray}
\LLL_{(d,p)}&=&\epsilon_{\underbrace{\cdot~\cdots~\cdot}_{d}}\underbrace{\RRR^\cdott\wedge\cdots\wedge\RRR^\cdott}_{p}\wedge
\underbrace{\eee^\cdot\wedge\cdots\wedge\eee^\cdot}_{d-2p},\label{LL2}
\end{eqnarray}
where $\epsilon_{\cdot~\cdots~\cdot}$ is a completely anti-symmetric tensor in $d$ dimensions with $\epsilon_{0\cdott d}=+1$.
\end{remark}
Even for a space-time manifold comprising torsions, a Lagrangian defined by (\ref{LL2}) is referred as a Lovelock Lagrangian in this report.
Among several Lovelock Lagrangians, the following three types are important:
\begin{enumerate}
\item the Lovelock Lagrangian of type $\LLL_{(d,1)}$, which is referred to as the ``Einstein--Hilbert Lagrangian''.
\item the Lovelock Lagrangian of type  $\LLL_{(2n,n)}$ which only includes the curvature form $\RRR$.
\item the Lovelock Lagrangian with two vielbein forms as $\LLL_{(2n,n-1)}$.
\end{enumerate}
The type $1$ Lagrangian, i.e., the Einstein--Hilbert Lagrangian, is certainly important because it describes  gravity on a four-dimensional spacetime.
Next the other two types will be discussed based on how the two types are special.

If the Lovelock Lagrangian can be constructed using the characteristic classes on the Poincar\'{e} bundle, a topological invariant can be obtained due to cohomological classes of the Lagrangian.
In such Lagrangian, a local $SO(1,d-1)$ symmetry may function as a gauge symmetry in the Yang--Mills gauge theory.
In reality, it is known that a gravitational constant, $\kappa$, in certain Lovelock Lagrangians can be discretitized using a topological invariant on the Poincar\'{e} bundle for an odd-dimensional spacetime\cite{PhysRevD.51.490}
\footnote{Zanelli\cite{PhysRevD.51.490} defined a ``Euler class'' as $\hat{\RRR}=\RRR+l^{-2}\eee\wedge\eee$, where $l$ is a dimensional constant. While a Euler class is a polynomial of curvature forms, ``Euler class'' defined by Zanelli is a mixture of a connection $\eee$ and a curvature $\RRR$. In this report, {\bf Theorem \ref{zanelli}} is stated based on Pontryagin and Chern classes.}.
This result is explained by the following theorem:
\begin{theorem}{\bf{(Zanelli)}}\label{zanelli}
On odd-dimensional torsionless space-time manifolds, the Lovelock actions of type $\LLL_{(2n+1,n)}$ are divisible modulo $\R/\Z$.
\end{theorem}
\begin{proof}
Let $\M_{2n-1}$ be a compact and simply connected $(2n-1)$-dimensional space-time manifold embedded in a $2n$-dimensional torsionless space-time manifold, $\M_{2n}$, with no boundary.
In addition, let $\tilde{\M}_{2n}\subset\M_{2n}$ be a $2n$-dimensional submanifold of $\M_{2n}$ whose  boundary is $\partial\tilde{\M}_{2n}=\M_{2n-1}$. 
The Lagrangian
\begin{eqnarray}
\LLL_{(2n,n)}&=&\frac{1}{(-2\pi)^nn!}\epsilon_{\underbrace{\cdot~\cdots~\cdot}_{2n}}~\underbrace{\RRR^\cdott\wedge\cdots\wedge\RRR^\cdott}_{n}=p_n
\end{eqnarray}
is the Lovelock Lagrangian of (\ref{LL2}).
A coefficient $a_p$ is chosen to create a $n$'s Pontryagin class of $p_n$.
On the torsionless manifold, the $n$'th Chern class of the Poincar\'{e} bundle becomes
\begin{eqnarray}
c_{2n}(\RRR)&=&\left(-1\right)^{n}\LLL_{(2n,n)}
\end{eqnarray}
due to the relation between the Chern and Pontryagin classes.
If $\RRR$ and $\RRR'$ are two different curvature forms in a $2n$-dimensional spacetime, the difference between the two Chern classes for $\RRR$ and $\RRR'$ is an exact due to the {\bf Theorem \ref{CW}-2}. 
Thus, this difference can be written as $c_{2n}(\RRR)-c_{2n}(\RRR')=d\CCC$.
The difference between two action integrals can be written as follows:
\begin{eqnarray}
\delta\I=\I(\RRR)-\I(\RRR')&=&\frac{\left(-1\right)^{n}}{\kappa}\int_{\tilde{\M}_{2n}}\left(
c_{2n}(\RRR)-c_{2n}(\RRR')\right)
,\nonumber\\
&=&\frac{\left(-1\right)^{n}}{\kappa}\int_{\tilde{\M}_{2n}}d\CCC
=\frac{\left(-1\right)^{n}}{\kappa}\int_{\M_{2n-1}}\CCC.
\end{eqnarray}
The integral in the last line represents the degree of the de Rham class on $\MM_{2n-1}$.
Thus, it is an integer.
This must be true for any solution of the Einstein equation.
Hence, action integrals must be divisible modulo $\R/\Z$.
\QED
\end{proof}
For {\bf Theorem \ref{zanelli}}, a torsionless condition is necessary to ensure that the Chern and Pontryagin classes are constructed using the Poincar\'{e} connection. 
If the space-time manifold comprises torsions, the Chern and Pontryagin classes should include the torsion forms because the Poincar\'{e}-curvature form includes the torsion form as shown in (\ref{curvf}).
The torsionless condition is necessary for {\bf Theorem \ref{zanelli}} on space-time manifolds of dimensions grater than three.
If space-time manifolds have a non-vanishing torsion, the Lovelock Lagrangian has no local translational invariance in general.
The only exception to this condition is the three-dimensional spacetime\cite{Witten198846}.
\begin{remark}{\bf (Witten)}:
The Einstein--Hilbert Lagrangian $\LLL_{(d,1)}$ and its translation belong to the same de Rham cohomology, if and only if the dimensionality of the space-time manifold is three.
\end{remark}
\begin{proof}
An infinitesimal translation operator, $\delta_T$, on the fundamental form can be expressed as follows\footnote{See, for instance, a section 4.1.1 of \cite{0264-9381-29-13-133001}}
\begin{eqnarray}
\delta_T \eee^a&=&d\lambda^a+\www^a_{~\cdot}\wedge\lambda^\cdot=\td\lambda^a,\\
\delta_T \www^{ab}&=&0,
\end{eqnarray}
where $\lambda^a$ is a translation parameter.
Therefore, the $d$-dimensional Einstein--Hilbert Lagrangian is transformed by the translation operator as
\begin{eqnarray}
\delta_T\LLL_{(d,1)}&=&\delta_T\left(\epsilon_{\underbrace{\cdot~\cdots~\cdot}_{d}}~\RRR^\cdott\wedge
\underbrace{\eee^\cdot\wedge\dots\wedge\eee^\cdot}_{d-2}\right),\nonumber\\
&=&\td\left((d-2)\epsilon_{\underbrace{\cdot~\cdots~\cdot}_{d}}~\RRR^\cdott
\wedge\lambda^\cdot\wedge\underbrace{\eee^\cdot\wedge\dots\wedge\eee^\cdot}_{d-3}\right)\nonumber\\
&~&-
(d-2)\epsilon_{\underbrace{\cdot~\cdots~\cdot}_{d}}~\RRR^\cdott
\wedge\lambda^\cdot\wedge\td(\underbrace{\eee^\cdot\wedge\dots\wedge\eee^\cdot}_{d-3}).
\end{eqnarray}
where a Bianchi identity, $\td\RRR=0$, is used.
Note that the covariant and external derivatives are identical to each other for local scalar functions.
Therefore, in the first term of the last line, a covariant derivative can be replaced by an external derivative, $\td(\bullet)\rightarrow d(\bullet)$.
The second term vanishes even for the space-time manifold comprising torsions if and only if the dimensionality of the space-time manifold is three.
Furthermore, the second Chern class in a three-dimensional spacetime can be written as follows:
\begin{eqnarray}
c_2(\FFF_0)&=&
\frac{1}{8\pi^2}\epsilon_{\cdots}~\RRR^\cdott\wedge\td\eee^\cdot
=\frac{1}{8\pi^2}d\LLL_{(3,1)}.
\end{eqnarray}
Therefore, 
\begin{eqnarray}
d\left(
\LLL_{(3,1)}-\delta_T\LLL_{(3,1)}
\right)&=&8\pi^2c_2(\FFF_0).
\end{eqnarray}
\QED
\end{proof}
Base on the above proof, it is also possible to show that  Lovelock Lagrangians of type $\LLL_{(2n+1,n)}$ are translation invariant up to a total derivative, even for $n\geq2$, if the requirement on the Lagrangian is relaxed to allow a non-Einstein--Hilbert type. 
This is because there is only one vielbein form in the Lagrangian.
This observation suggests a new method of constructing translationally invariant Lagrangians (up to a total derivative) of the Einstein--Hilbert type on an even-dimensional space-time manifold.

%%%%%%%%%%%%%%%%%%%%%%%%%%%%%%%%
% Section 3: co-translation operator
%%%%%%%%%%%%%%%%%%%%%%%%%%%%%%%%
\section{Co-translation operator and co-Poincar\'{e} bundle}\label{cotrans}
Let $P_a$ be a generator of translation group $T^d$ on a $d$-dimensional local space-time manifold, $\M$.
A ``co-translation'' generator is introduced as follows:
\begin{definition}{\bf (co-translation)}: A co-translation $P_{ab}$ can be expresses as follows: 
\begin{eqnarray}
P_{ab}&=&P_a\iota_b,
\end{eqnarray}
where $\iota_\bullet$ is a contraction with respect to a local frame vector, $\xi^\bullet$\@.
A representation of $\iota_\bullet$ by a global frame bundle can be expressed as follows:
\begin{eqnarray}
\iota_{a}&=&\iota_{\xi^a},~~\xi^a=\eta^{ab}\Varepsilon^{\mu}_b\partial_\mu.
\end{eqnarray}
\end{definition}
A contraction is a map $\iota:\bigwedge^p \rightarrow \bigwedge^{p-1}$, where $\bigwedge^r$ is a space of $r$-forms on $T^*\M$.
One can obtain a relation,
\begin{eqnarray}
\iota_{a}\eee^b&=&\Varepsilon^{\mu}_a\Varepsilon_\nu^b\delta^\nu_\mu=\delta^b_a,
\end{eqnarray}
which is independent of the choice of frame bundles.

A projection manifold and a projection bundle can be introduced  using a contraction operator as follows:
let $\M_{\perp a}\subset\M$ be a $(d-1)$-dimensional submanifold of a $d$-dimensional space-time manifold, $\M$, such that $\iota_a\aaa=0$ for any one-forms $\aaa\neq0$ on $T^*\M_{\perp a}$.
The trivial frame bundles on $T\M_{\perp a}$ and $T^*\M_{\perp a}$ can be regarded as sub-bundles of $T\M$ and $T^*\M$, respectively.
If a local space-time manifold, $\M$, has the Poincar\'{e} symmetry $ISO(1,d-1)$, the submanifold $\M_{\perp a}$ has the $ISO(1,d-2)$ symmetry or the $ISO(d-1)$ symmetry.
Therefore, the connection bundle $\AAA_a$ and the curvature bundle $\FFF_a=d\AAA_a+\AAA_a\wedge\AAA_a$  can be defined on $\M_{\perp a}$ as sub-bundles of $\AAA$ and $\FFF$ on $\M$, respectively.
Using an equivalence relation $\sim_a$ such as $\aaa\sim_a\bbb \Leftrightarrow \iota_a\aaa=\iota_a\bbb$, quotient bundles $\widetilde{\AAA}_a=\AAA/{\sim_a}$ and $\widetilde{\FFF}_a=\FFF/{\sim_a}$ can be introduced, where $\aaa$ and $\bbb$ are any $p$-forms on $T^*\M$.
It is clear that $\AAA_a\simeq\widetilde{\AAA}_a\subset\AAA$ and $\FFF_a\simeq\widetilde{\FFF}_a\subset\FFF$.

The Lie algebra of the co-Poincar\'{e} group is obtained as follows:
\begin{eqnarray}
\left[P_{ab},P_{cd}\right]&=&0,\label{PJcr4}\\
\left[J_{ab},P_{cd}\right]&=&-\eta_{ac}P_{bd}+\eta_{bc}P_{ad},\label{PJcr5}
\end{eqnarray}
and the Lie algebra (\ref{PJcr3}).
According to this Lie algebra, co-Poincar\'{e} connection and curvature bundles can be constructed as follows:
\begin{eqnarray}
\AAA&=&\{J_{ab}, P^{cd}\}\times\{\vomega^{ab},\SSS_{cd}\}
=J_\cdott\vomega^\cdott+ P^\cdott\SSS_\cdott,\label{cPB1}\\
\FFF&=&d\AAA+\AAA\wedge\AAA=J_\cdott\RRR^\cdott+P^\cdott\td\SSS_\cdott,\label{cPB2}
\end{eqnarray}
where $\SSS_{ab}=\epsilon_{ab\cdott}\eee^\cdot\wedge\eee^\cdot/2$ is a surface form, i.e., a two-dimensional surface perpendicular to a plane spanned by $\eee^a$ and $\eee^b$.
The lowering and raising of the Roman indexes are done using the Lorentz metric $\eta$.
This choice of the co-Poincar\'{e} bundles corresponds to the isomorphism $SO(4)\simeq SU(2)\times SU(2)$ and is not unique.
There is another isomorphism, $SO(2,2)\simeq SL(2,R)\times SL(2,L)$, which gives another Lie algebra $\left[P_{ab},P_{cd}\right]\neq0$.
The latter case is not considered in this report.

The co-translation induces an infinitesimal co-translation operator, $\delta_{CT}=\xi^a\times\delta_T\iota_a$, on the fundamental forms as follows:
\begin{eqnarray}
\delta_{CT}\left(\eee^a\wedge\eee^b\right)&=&
\xi^a\times\td\xi^b-\xi^b\times\td\xi^a,\label{dctS}\\
\delta_{CT} \www^{ab}&=&0,\label{dctW}
\end{eqnarray}
where the translation parameter is considered as a local frame vector, $\lambda^\bullet=\xi^\bullet$.
Here, $\xi^a\times\xi^b=-\xi^b\times\xi^a$ is an outer product of the two frame vectors.
For both the translation and contraction operators, it is necessary to specify the directions of the translation and contraction.
While formula (\ref{dctS}) is expressed using the local and global frame vectors, all the following results are independent of the choice of frame bundles.
For a $\LLL_{(2n,n-1)}$-type Lovelock Lagrangian, the following remark can be stated:
\begin{remark}{\bf The co-Poincar\'{e} symmetry of the Lovelock action $\LLL_{(2n,n-1)}$}{\rm :}
The Lovelock Lagrangian $\LLL_{(2n,n-1)}$ is invariant under the co-Poincar\'{e} transformation up to a total derivative.
\end{remark}
\begin{proof}
As the invariance under the local $SO(1,2n-1)$ symmetry of the Lovelock Lagrangian is trivial, a proof of co-translation invariance is given here.
A co-translation operator $\delta_{CT}$ induces a transformation on the Lovelock Lagrangian $\LLL_{(2n,n-1)}$ as
\begin{eqnarray}
\delta_{CT}\LLL_{(2n,n-1)}&=&\delta_{CT}\left(
\epsilon_{\underbrace{\cdot~\cdots~\cdot}_{2n}}\underbrace{\RRR^\cdott\wedge\cdots\wedge\RRR^\cdott}_{n-1}\wedge
\eee^\cdot\wedge\eee^\cdot
\right),\nonumber\\
&=&2
\epsilon_{\underbrace{\cdot~\cdots~\cdot}_{2n}}\underbrace{\RRR^\cdott\wedge\cdots\wedge\RRR^\cdott}_{n-1}
\wedge(\xi^\cdot\times\td\xi^\cdot)\label{CT1}
\end{eqnarray}
using (\ref{dctS}) and (\ref{dctW}).
Here, note that the relation
\begin{eqnarray}
\epsilon_{\cdots ab}\wedge\td(\xi^a\times\xi^b)&=&\epsilon_{\cdots ab}\wedge(\td\xi^a\times\xi^b)+
\epsilon_{\cdots ab}\wedge(\xi^a\times\td\xi^b),\nonumber\\
&=&2\epsilon_{\cdots ab}\wedge(\xi^a\times\td\xi^b).
\end{eqnarray}
Thus,
\begin{eqnarray}
(\ref{CT1})&=&
\td\left(
\epsilon_{\underbrace{\cdot~\cdots~\cdot}_{2n}}\underbrace{\RRR^\cdott\wedge\cdots\wedge\RRR^\cdott}_{n-1}
(\xi^\cdot\times\xi^\cdot)
\right),
\end{eqnarray}
where the Bianchi identity $\td\RRR=0$ is used.
Because $\td(\bullet)$ can be replaced with $d(\bullet)$ for a local scalar, the co-translation operator only induces a total derivative term on the Lovelock Lagrangian $\LLL_{(2n,n-1)}$.
\QED
\end{proof}
When a $\LLL_{(2n,n-1)}$-type Lovelock Lagrangian is given on a local $2n$-dimensional space-time manifold $\M_{2n}$, the co-translation operator induces a short exact sequence as follows:
\begin{eqnarray}
0\rightarrow\LLL_{(2n,n-1)}\rightarrow \LLL_{(2n-2,n)}\rightarrow\LLL_{(2n-2,n-1)}/\LLL_{(2n,n-1)}\rightarrow0,\label{ses1}
\end{eqnarray}
where $\LLL_{(2n-2,n)}$ is the Lovelock Lagrangian defined on a submanifold of $\M_{\perp a}\cap\M_{\perp b}\subset\M_{2n}$.
Note that $\delta_{CT}\LLL_{(2n-2,n-1)}=0$ based on (\ref{dctW}).
According to the choice of $a$ and $b<a$, the Lovelock Lagrangian of $\LLL_{(2n-2,n-1)}$ has $n(2n-1)$ splittings in total.
Each of splitting corresponds to a conserved Noether current.

%%%%%%%%%%%%%%%%%%%%%%%%%%%%%%%%
% Section 4: Einstein-Hilbert Lagrangian in four dimension
%%%%%%%%%%%%%%%%%%%%%%%%%%%%%%%%
\section{Four dimensional Einstein--Hilbert Lagrangian}\label{EH}
It is clearly evident that the Einstein--Hilbert Lagrangian $\LLL_{(d,1)}$ has the co-Poincar\'{e} symmetry, if and only if it is in a four-dimensional spacetime:
\begin{eqnarray*}
\LLL_{(2n,n-1)}=\LLL_{(d,1)}&\Longleftrightarrow& n=2,~d=4.
\end{eqnarray*}
The Einstein--Hilbert Lagrangian $\LLL_{(4,1)}$ is nothing more than a Lagrangian of general relativity in four dimensions for pure gravity without any matter fields and a cosmological term.
For $\LLL_{(4,1)}$, the following theorem can be stated:
\begin{theorem}{\bf (Chern classes in general relativity)}\label{YK}
The Einstein--Hilbert Lagrangian is equivalent to the $r$'th Chern class with respect to the co-Poincar\'{e} bundle, if and only if it is defined in a four-dimensional space-time manifold embedded in a five-dimensional space-time manifold with $r=2$.
\end{theorem}
\begin{proof}(sufficiency):
Let $\M_{4}$ be a compact and simply connected four-dimensional space-time manifold embedded in a five-dimensional space-time manifold, $\M_{5}$, with no boundary, i.e., $\partial\M_5=0$.
Furthermore, let $\tilde{\M}_{5}\subset\M_{5}$ be a five-dimensional submanifold of $\M_{5}$ whose boundary is $\partial\tilde{\M}_{5}=\M_{4}$. 
The second Chern class with respect to the co-Poincar\'{e} curvature bundle (\ref{cPB2}) on $\M_5$ is expresses as follows:
\begin{eqnarray}
c_2(\FFF)&=&\frac{1}{8\pi^2}Tr\left[\FFF\wedge\FFF\right]=\frac{1}{8\pi^2}\td\SSS_\cdott\wedge\RRR^\cdott
=\frac{1}{8\pi^2}d\left(\epsilon_{\cdott\cdott}\RRR^\cdott\wedge\eee^\cdot\wedge\eee^\cdot\right).\label{YK2}
\end{eqnarray}
Note that $d(\bullet)$ is the external derivative on a five-dimensional space-time manifold.
This shows that the second Chern class has a closed form, which is a concrete example of {\bf Theorem \ref{CW}-1}.
Thus, the Einstein--Hilbert action can be written as follows:
\begin{eqnarray}
\I_4&=&\frac{1}{\kappa}\int_{\M_4}\LLL_{(4,1)}
=\frac{1}{\kappa}\int_{\tilde{\M}_5}d\left(\epsilon_{\cdott\cdott}\RRR^\cdott\wedge\eee^\cdot\wedge\eee^\cdot\right)
=\frac{8\pi^2}{\kappa}\int_{\tilde{\M}_5}c_2.
\end{eqnarray}
Consequently, the sufficiency case of the theorem has been proved.
\\
(necessity): The Einstein--Hilbert Lagrangian can be defined using $\RRR$ and $\SSS$ as
\begin{eqnarray}
\LLL_{(d=2n,1)}&=&\epsilon_{\underbrace{\cdot~\cdots~\cdot}_{2n}}~\RRR^\cdott\wedge
\underbrace{\SSS^\cdott\wedge\dots\wedge\SSS^\cdott}_{n-1}
\end{eqnarray}
only when the space-time dimension $d$ is an even number.
Further, the $r$'th Chern class, including one $\RRR$ on a $2n$-dimensional manifold, is constructed using the terms
$
\RRR^{a_1}_{~a_2}\wedge\td\SSS^{a_2}_{~a_3}\wedge\cdots\wedge\td\SSS^{a_p}_{~a_1},
$
where $2\leq p\leq r\leq 2n$.
By comparing these terms with $\LLL_{(2n,1)}$, it is obvious that $c_r=d\LLL_{(2n,1)}$ can be realized only with $n=2$ and $r=2$.
\QED
\end{proof}
In this case, the second Chern class can be recognized as a conserved Noether charge due to {\bf Remark 3}.
As a special case of (\ref{ses1}),  the short exact sequence based on {\bf Theorem \ref{YK}} can be obtained as follows:
\begin{eqnarray}
0\rightarrow c_2(\FFF)~{\rm on}~\M_5 \rightarrow \LLL_{(4,1)}\rightarrow H^*_D(\FFF) \rightarrow0,\label{ses1}
\end{eqnarray}
where $H^*_D(\FFF)$ is the de Rham cohomology on the co-Poincar\'{e} curvature bundle $\FFF$ on the five-dimensional space-time manifold.
From this observation, the following theorem can be stated:
\begin{theorem}{\bf (cohomological class of the Einstein--Hilbert action)}
On a four-dimensional space-time manifold, the Einstein--Hilbert Lagrangian $\LLL_{(4,1)}$ is divisible modulo $\R/\Z$.
\end{theorem}
\begin{proof}
The proof is the same as that of {\bf Theorem \ref{zanelli}}.
\QED
\end{proof}
\begin{example}{\bf (cosmological constant in de Sitter spacetime)}:\label{ex1}
On a four-dimensional space-time manifold, the Einstein--Hilbert-type Lovelock Lagrangian with a cosmological term can be written as follows:
\begin{eqnarray}
\LLL_{\Lambda}&=&\frac{1}{4}\left(\LLL_{(4,1)}+\frac{\Lambda}{3!}\LLL_{(4,0)}\right),
\end{eqnarray}
where $\Lambda$ is the cosmological constant.
A solution of the Einstein equation induced by $\LLL_\Lambda$ is obtained as follows:
\begin{eqnarray}
\eee^a&=&\left(f(r)dt,f^{-1}(r)dr,rd\theta,r\sin{\theta}d\phi\right),\label{dSsol}
\end{eqnarray}
where $f(r)=\sqrt{1+\Lambda r^2/3}$.
Herein, this solution is referred to as the de Sitter (dS) solution, whether the cosmological constant is positive or negative.
While the Lagrangian $\LLL_\Lambda$ itself is not a Chern class, the dS solution induces the topological invariant.
When the vielbein form (\ref{dSsol}) is substituted to (\ref{YK2}), a topological invariance can be obtained as
\begin{eqnarray}
d\left(\epsilon_{\cdott\cdott}\RRR^\cdott\wedge\eee^\cdot\wedge\eee^\cdot\right)&=&8\Lambda~d\vvv=8{\pi^2}c_2,
\end{eqnarray}
where $\vvv=\epsilon_{\cdott\cdott}\eee^\cdot\wedge\eee^\cdot\wedge\eee^\cdot\wedge\eee^\cdot/4!$ is the volume form of a four-dimensional space-time manifold.
Assume that the four-dimensional dS space $\M_4$ is embedded in a five-dimensional space-time manifold, $\M_5$, with no boundary, i.e., $\M_4=\partial(\tilde{M}_5\subset\M_5)$.
With this assumption, the following relation can be obtained:
\begin{eqnarray}
\Lambda\int_{\M_4}\vvv&=&{\pi^2}\int_{\tilde{\M}_5}c_2.
\end{eqnarray}
Therefore, based on  {\bf Theorem \ref{YK}},  the cosmological constant is divisible modulo $\R/\Z$.
\end{example}

While only the pure gravitational Lagrangian without matter fields has been discussed in this report so far, topological properties of the gravitational part may affect matter fields, which are the sources of gravity.
Let us now discuss a possible effect of the topological properties of general relativity on matter fields in the simplest case.
Let $\varphi$ be a real  scalar field, which is a map $\varphi:\M\rightarrow\R:\xi\mapsto\varphi(\xi)$, in which $\xi$ is a local frame-vector field.
The field $\varphi$ is assumed to be differentiable an arbitrary number of times. 
The four-dimensional Lagrangian of the scalar field can be written on a local space-time manifold  as follows:
\begin{eqnarray}
\LLL_s&=&\frac{1}{2}
\left(
\eta^\astt\iota_\ast\sss^\cdot\wedge\iota_\ast\sss^\cdot-\frac{1}{3!}V(\varphi)\eee^\cdot\wedge\eee^\cdot
\right)\wedge\SSS_\cdott,\label{Ls}
\end{eqnarray}
where $V(\varphi)(\xi)$ is the potential energy (functional of $\varphi(\xi)$) and $\sss^\bullet(\xi)$ is the scalar-field two-form defined as
\begin{eqnarray}
\sss^a&=&d\varphi\wedge\eee^a=\frac{\partial\varphi}{\partial\xi^\cdot}\eee^\cdot\wedge\eee^a.
\end{eqnarray}
The first term  in a parenthesis in (\ref{Ls}) (kinetic term) can be represented as $g^{\mu\nu}\partial_\mu\varphi~\partial_\nu\varphi$ using  global frame vectors.
In general, the Lagrangian of the scalar field has no topological invariance.
The total Lagrangian is given by the sum of the gravitational and scalar-field Lagrangians as follows:
\begin{eqnarray}
\LLL_{{\rm tot}}&=&\LLL_G+\LLL_s=
\frac{1}{2\kappa}\RRR^\cdott\wedge\SSS_\cdott
+\frac{1}{2}\eta^\astt\iota_\ast\sss^\cdot\wedge\iota_\ast\sss^\cdot\wedge\SSS_\cdott
-V(\varphi)\vvv.
\end{eqnarray}
The equations of motion can be obtained using a stationary condition on the action integral with respect to the virbein, spin and scalar-filed forms as 
\begin{eqnarray}
\delta_\aaa\I&=&\delta_\aaa\int\left(\LLL_G+\LLL_s\right)=0,
\end{eqnarray}
where $\aaa=\www$, $\eee$, or $\sss$.
Possible configurations of the scalar field and space-time manifold can be determined by simultaneously solving these equations  under specific boundary conditions.
Although the scalar field itself does not have any topological invariants, a configuration of the scalar field must inevitably receive some constraint from the space-time manifold. 
\begin{example}{\bf  Friedmann-Lme\^itre-Robertson-Walker metric based on a scalar field}:
The Friedmann-Lme\^itre-Robertson-Walker (FLRW) metric\cite{Einstein1922,Lmeitre,1935ApJ82284R,Walker01011937} is a homogeneous and isotropic solution of the Einstein equation given by
\begin{eqnarray}
\eee^a&=&\left(dt,a(t)f^{-1}(r)dr,a(t)rd\theta,a(t)r\sin{\theta}d\phi\right),\label{FLRW}
\end{eqnarray}
where $f(r)=\sqrt{1-K r^2}$ and $a(t)$ is a scale function.
The scale function $a(t)$ is given as a solution of the Friedmann equation
\begin{eqnarray}
\frac{\ddot{a}}{a}&=&-\frac{\kappa}{3}\left(
\rho+3P
\right),
\end{eqnarray}
where $\ddot{a}=d^2a(t)/dt^2$.
If there is only one scalar field on the space-time manifold, $\rho$ and $P$ are given by the energy-momentum tensor $T_\bullett$ of the scalar field  as follows:
\begin{eqnarray}
T_{ab}&=&
\partial_{a}\varphi~\partial_{b}\varphi
-\frac{1}{2}\eta_{ab}\partial_\cdot\varphi~\partial^\cdot\varphi
+\eta_{ab}V,\\
\rho&=&T^0_{~0},~P=\frac{T^1_{~1}}{3}\left(=\frac{T^2_{~2}}{3}=\frac{T^3_{~3}}{3}\right).
\end{eqnarray}
From the FLRW metric (\ref{FLRW}), the second Chern class can be calculated using (\ref{YK2}) as follows:
\begin{eqnarray}
d\left(\epsilon_{\cdott\cdott}\RRR^\cdott\wedge\eee^\cdot\wedge\eee^\cdot\right)
&=&\frac{1}{3}\left(\rho-3P\right)=8{\pi^2}c_2.
\end{eqnarray}
Note that $\rho-3P=\eta^{ab} T_{ab}$, and thus, it is independent of the choice of a global frame bundle.
Therefore, under the condition same as that of {\bf Example \ref{ex1}}, the trace of the energy-momentum tensor of the scalar field $Tr[T^\bullet_{~\bullet}]=\sum_a\eta^{a\cdot} T_{\cdot a}$ is divisible modulo $\R/\Z$.
Note that the trace of the energy-momentum tensor is invariant under the global $GL(1,3)$.
Explicit calculations can show a relation
\begin{eqnarray}
T_{ab}&\rightarrow&T_{\mu\nu}=
\partial_{\mu}\varphi~\partial_{\nu}\varphi
-\frac{1}{2}g_{\mu\nu}g^{\lambda_1\lambda_2}\partial_{\lambda_1}\varphi~\partial_{\lambda_2}\varphi
+g_{\mu\nu}V,\\
Tr[T^a_{~b}]&=&Tr[T^\mu_{~\nu}]=\sum_{\mu,\nu} g^{\mu\nu}T_{\nu\mu},
\end{eqnarray}
where a concrete form of $g^{\mu\nu}$ can be read out from (\ref{FLRW}).
\end{example}

%%%%%%%%%%%%%%%%%%%%%%%%%%%%%%%%
% Section 5: Summary
%%%%%%%%%%%%%%%%%%%%%%%%%%%%%%%%
\section{Summary}
The co-translation operator that performs successive operations of contraction and translation is introduced.
According to this new operation, the co-Poincar\'{e} connection and curvature bundles can be defined on the space-time manifold.
We show that some types of Lovelock Lagrangians are invariant up to a total derivative under co-translation.
Consequently, we obtain topological invariants in general relativity.
The main result of this study shows that a four-dimensional Einstein--Hilbert Lagrangian can be recognized as the second Chern class, which defines the de Rham cohomology in general relativity.
Moreover, the results show that the cosmological constant and trace of the energy-momentum tensor based on the scalar field are modulo $\R/\Z$ as examples of this theorem.
%%%%%%%%%%%%%%%%%%%%%%%%%%%%%%%%%%%%%%%%%%%%%%%%%%%%%%%%%%%%%%%%%%%%%%%%%%%%%%%%%%%%%%%%
%\begin{acknowledgments}
\vskip 5mm
\thanks {
I appreciate the kind hospitality of all members of the theory group of Nikhef, particularly Prof. J. Vermaseren and Prof. E. Laenen.
A major part of this study has been conducted during my stay at Nikhef in 2016.
In addition, I would like to thank Dr.~Y.~Sugiyama for his continuous encouragement and fruitful discussion.
}
%\end{acknowledgments}
%%%
%
% Appendix
%

%
% BibTeX users please use
%
%\bibliographystyle{elsarticle-num}
%\bibliographystyle{spphys}
\bibliographystyle{spmpsci}
\bibliography{ref}
\end{document}